\documentclass{sig-alternate}
\usepackage{multirow}
\usepackage{multirow}
\usepackage{graphicx}
\usepackage{subfigure}
\usepackage{amsfonts}
\usepackage{amsmath}
\usepackage{algorithmic}
\usepackage{algorithm}
\newtheorem{theorem}{Theorem}
\newdef{definition}{Definition}
\newtheorem{lemma}{Lemma}

\begin{document}

\conferenceinfo{xxxx,} {2012, Beijing, China.}
\CopyrightYear{2012}
\crdata{978-1-4503-0717-8/11/10}
\clubpenalty=10000
\widowpenalty = 10000

\title{Modeling Relational Data via Latent Factor Blockmodel}

\numberofauthors{1}
\author{
\alignauthor Sheng Gao, Ludovic Denoyer, Patrick Gallinari \\
       \affaddr{LIP6 - UPMC}\\
       \affaddr{4 place Jussieu, 75005 Paris, France}\\
       \email{ \{sheng.gao, ludovic.denoyer, patrick.gallinari \}@lip6.fr}
}
\maketitle

\begin{abstract}
In this paper we address the problem of modeling relational data, which appear in many applications such as social network analysis, recommender systems and bioinformatics. Previous studies either consider latent feature based models but disregarding local structure in the network, or focus exclusively on capturing local structure of objects based on latent blockmodels without coupling with latent characteristics of objects. To combine the benefits of the previous work, we propose a novel model that can simultaneously incorporate the effect of latent features and covariates if any, as well as the effect of latent structure that may exist in the data. To achieve this, we model the relation graph as a function of both latent feature factors and latent cluster memberships of objects to collectively discover globally predictive intrinsic properties of objects and capture latent block structure in the network to improve prediction performance. We also develop an optimization transfer algorithm based on the generalized EM-style strategy to learn the latent factors. We prove the efficacy of our proposed model through the link prediction task and cluster analysis task, and extensive experiments on the synthetic data and several real world datasets suggest that our proposed LFBM model outperforms the other state of the art approaches in the evaluated tasks.
\end{abstract}




\section{Introduction}
Nowadays, relational data has become ubiquitous, which consists of interrelated objects with multiple relation types, such as in online social networks people connect to each other by their friendship, or research papers can be connected by citation or co-authorship. Thus, modeling relational data has arisen as a fundamental task in many applications, which involves to predict the new relations among objects and to discover latent structure among the networked data. For instance, given the partially observed relational data from a social network, one may be interested to predict the missing relationship between unobserved pairs of individuals, or to identify the groups of users who share common interest in a particular product or service.

However, the complexity of relational data makes statistical modeling a great challenging task: First, the correlations among objects give rise to various structural patterns, which exhibits the property of \emph{stochastic equivalence} in the relational data \cite{Hoff:MLFM}. This characteristic implies that the objects can be divided into clusters where members within a cluster have similar pattern of relations to other objects, i.e., the cluster structure is either dense or sparse, which deviates the classical clustering assumption, i.e., the strongly correlated data always forms dense clusters \cite{Long:SCC}. Take an online social network as an example, people in the same company can form dense circles due to their professional relationships, while others can also constitute a sparse group where they share the same preference for buying some product in \texttt{Groupon} but may not be linked to each other, which is demonstrated in Figure 1.
Second, relational data is quite sparse, because each graph generated from the relational data involves a number of objects with each being connected to only a tiny proportion of the whole graph, which calls for the statistical modeling capable of learning from rare, noisy and largely missing observations.
Third, in addition to containing the structure information the relational data may have extra side information specific to the objects. Thus, encoding heterogenous information sources in the relational data is required for a flexible modeling. Finally, large scale relational data in many real world applications ask for the statistical modeling efficiently scalable.

Previous work can be classified into \emph{feature} based and \emph{structure} based models. The feature based models employ latent matrix factorization framework to learn latent factors for each object and each relation, and make predictions by taking appropriate inner products. Their strength lies in the relative ease of their continuous optimization and in their excellent predictive performance. The representative model is the Multiplicative Latent Factor Model \cite{Hoff:MLFM}, which associates with each user a low dimensional latent sender factor and latent receiver factor. However, this approach disregards the local structure among the relational data due to other latent unmeasured factors, and also lacks interpretable representations for the latent structure. In contrast, structure based models focus exclusively on capturing latent structures in the relational data \cite{ai:mmsb}. The discovered latent structures from such models can provide insights about the interactions in the relational data which are useful in the absence of side information. In fact, these latent structures provide a parsimonious model to capture the interactions among objects. However, since this kind of approaches do not adjust for the effects of objects' features, the resulting latent structure may contain redundant information.

\begin{figure*}[t]
\centering
\includegraphics[width=5in]{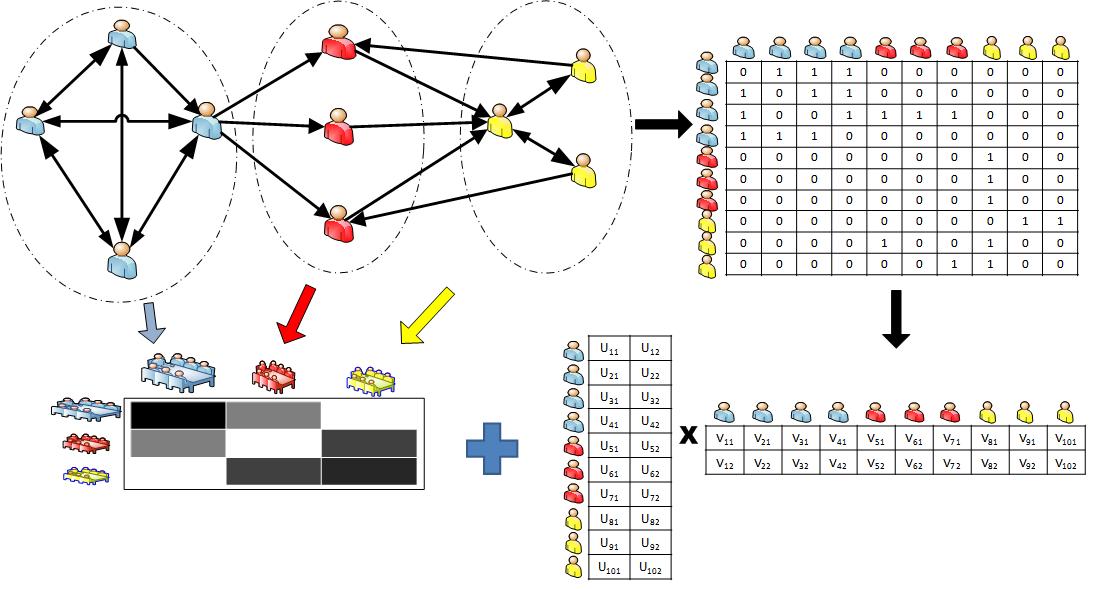}
\caption{Modeling the relational data in the context of a social network. }
\end{figure*}

In this paper, we propose a statistical model that can incorporate the effects of latent features and covariates if any, and also the impacts from any latent structure that may exist in the relational data simultaneously. To achieve this, we model the objects and relations from the observational data as a function of both latent feature factors and latent structural factors to collectively discover globally predictive intrinsic properties of objects and capture the latent structure in the data. More precisely, we exclusively assign each object to one and only one latent cluster, which partitions the relation matrix into a small number of blocks or co-clusters. In this case, the estimated latent feature factors of objects and the corresponding probabilities of being clustered into certain latent blocks are considered to be the representative of latent factors that contribute to the interactions in the relation matrix, the procedure of modeling the relational data is shown as in Figure 1. By coupling these latent factors, our proposed model can not only provide better predictive performance but also discover the interpretable latent block structure.

\emph{Contributions:} This paper provides a predictive modeling approach for the relational data, which can integrate the information sources from the features as well as the local structures in the relation matrix. Specifically, our work in this paper makes the following several contributions:
\begin{enumerate}
    \item We present a novel approach to model the relational data as a function of latent feature factors and latent block structural factors through our proposed Latent Factor Block Model (LFBM).
    \item To efficiently learn the LFBM model, we propose an optimization transfer algorithm based on the Generalized Expectation-Maximization (EM) method, so-called Minorization-Maximization (MM) algorithm to inference the latent factors and model parameters.
    \item Extensive experiments are conducted on the synthetic data and several real world datasets. The proposed model is shown to outperform state-of-the-art methods in terms of the better predictive performance and clearer latent cluster structures.
\end{enumerate}

The paper is structured as follows. We first briefly introduce the problem definition in Section 2. The proposed framework based on the Latent Factor Block Model and the model specification are presented in Section 3. We derive the efficient optimization algorithm to learn the model in Section 4. Then we describe experiments on the synthetic data as well as several real world datasets, and provide comparisons with state-of-the-art methods in Section 5. Section 6 is about the related work. In Section 7 we present conclusions and future work.

\section{Problem Definition}
Before introducing our proposed statistical model, we first give the notations that are used throughout this paper.
Suppose we have a set of objects $\{ x_{1}, ..., x_{n}\}$. Observations consisting of pairwise measurements are respectively represented by the relation matrix $\mathbf{S} = \{ S_{ij} \in \{ 0, 1, ? \} , i,j= 1,...,N \}$, where $1$ denotes there is an observed present relation, $0$ denotes the absent relation, and $?$ denotes the missing relation. We then use the binary indicator matrix $\mathbf{W}$ to indicate whether or not the relation is observed, More specifically, $W_{ij}=1$ means that $S_{ij}$ is observed while $W_{ij}=0$ means $S_{ij}$ is missing. We use ${z}_{i}\in \{ 1,...,K\}$, where $K$ is the number of latent clusters in the relational data, to denote the cluster assignment of object $x_{i}$ and we refer to ${z}_{i}$ as the latent cluster membership of object $x_{i}$. We furthermore introduce $z_{ik}= [z_{i} = k] $ to indicate that $x_{i}$ is in the $k$th cluster when $z_{ik} = 1$ and 0 otherwise. Latent cluster assignments matrix $\mathbf{Z} = \{z_{ik} : i \in {1,..., N}, k \in {1,...,K}\}$ includes the latent cluster memberships of all the objects in the relation graph. Given such a relation graph, our goal is to predict the missing relations between the unobserved pairs and to discover the latent structure among the objects as well.

\section{Our Proposed Model}
We consider modeling the relational data based on the latent factor models. We first define a Bernoulli-Logistic generative process where the interactions among objects are generated, and then propose a hybrid model to capture the latent features factors and latent structural information for each object, which combine the benefits of both feature based and structure based models.

\subsection{Bernoulli-Logistic Based Model}
In case of the relational data, We first assume the elements of the relation matrix ${S}_{ij}$ as Bernoulli-distributed variables, which are conditionally independent given the latent variables ${H}_{ij}$ through the logistic function $ \sigma (H) = \frac{1}{1 + e^{- H}} $. Thus, the likelihood of the observations in the relation matrix $\mathbf{S}$ can be defined as follows:
\begin{equation}
\label{slike}
\begin{aligned}
P( \mathbf{S} | \mathbf{H} ) = \prod_{i,j} [\sigma(H)_{ij}^{S_{ij}} (1- \sigma (H)_{ij})^{1- S_{ij}}]^{W_{ij}}
\end{aligned}
\end{equation}

To characterize the latent variable ${H}_{ij}$ in the framework of latent feature models, we consider that there exists \emph{latent feature factor} $\mathbf{u}_{i} \in  \mathbb{R}^{d}$ for object $x_{i}$, where $d$ is the dimension of latent feature factor, that could be used for encoding the observable attributes (e.g. a user's profile) or latent semantic topics (e.g. a movie's genre). Furthermore, based on the latent cluster membership ${z}_{i}$, we introduce a latent block matrix $\mathbf{C} \in  \mathbb{R}_{+}^{K \times K}$ to explicitly capture the latent local structure, where $C_{kk}$ denotes the probability of a link existing between objects within the same $k$th cluster, and $C_{kl}$ denotes the probability of one object in $k$th cluster linking to the other object within the $l$th cluster. Then the latent variable ${H}_{ij}$ can be defined as follows:
\begin{equation}
\label{logodd}
\begin{aligned}
H_{ij}  = \mathbf{u}_{i} \mathbf{v}_{j}^{T} + \mathbf{z}_{i} \mathbf{C} \mathbf{z}_{j}^{T} + \epsilon
\end{aligned}
\end{equation}

Here, $\mathbf{u}_{i} \in  \mathbb{R}^{d}$ and $\mathbf{v}_{i} \in  \mathbb{R}^{d}$ denotes latent sender feature vector and latent receiver feature vector of each object \footnote{Note that for directed interaction matrix we use different latent feature vectors for the same object as used in MLFM model \cite{Hoff:MLFM}.}. Then the inner product model of $\mathbf{u}_{i} \mathbf{v}_{j}^{T}$ provides the probabilities of a relation between the two objects based on their latent features.

Then, $\mathbf{z}_{i} \in  \mathbb{R}^{K}$ is considered as the latent cluster indicator \emph{vector} for each object, which implies the object $x_{i}$ associates with the $k$th cluster. Actually, the form of $\mathbf{z}_{i} \mathbf{C} \mathbf{z}_{j}^{T}$ provides a general model to discover different latent cluster structure in the relational data, i.e., dense or sparse cluster. More specifically, we can use the block matrix $\mathbf{C}$ to represent various types of the latent cluster structures in the relational data. For example, the model can learn only spare clusters by fixing diagonal elements of $\mathbf{C}$ to be zeros, while it can also find the dense clusters by fixing $\mathbf{C}$ as the identity matrix, which means that the block structure provides a principal way to adapt the model to learn specific types of latent cluster structures. $\epsilon$ denotes the sparsity of the relations in the network, which can also be considered as a kind of bias term \footnote{For computation convenience in our experiments, the bias $\epsilon$ is absorbed into latent factors $U$ and $V$ by redefinition.}.

More generally, considering there also exists side information or covariates about the objects or the interactions, we can easily incorporate them into the proposed model in the similar way as in Generalized Linear Models \cite{Agarwal:PDLFM}. Hence, the latent interaction matrix can be defined as follows:
\begin{equation}
\label{logodd2}
\begin{aligned}
H_{ij}  = \beta^{T} \mathbf{x}_{ij} + \mathbf{u}_{i} \mathbf{v}_{j}^{T} + \mathbf{z}_{i} \mathbf{C} \mathbf{z}_{j}^{T} + \epsilon
\end{aligned}
\end{equation}
where the vector $\mathbf{x}_{ij}$ represents the side information about the relation between the objects $i$ and $j$, $\beta $ denotes the regression coefficients associated with the pre-defined side information, which could be assigned a normal prior distribution.

\subsection{Model Specification}
We thus far model the observed interactions by combining the benefits of latent feature factors of objects with their corresponding latent cluster assignments as well as latent block structure, and the integration of multiple effects makes our proposed model better generalization in link prediction and more interpretable for network structure, which is referred to as Latent Factor BlockModel (LFBM). The LFBM is different from the classical factorization-style link prediction model which disregards the local latent structure, and also differs from the traditional clustering based models that are not flexible to account for the side information and process the missing relations in the data.

Moreover, the factorization based latent feature selection term $\mathbf{u}_{i} \mathbf{v}_{j}^{T}$ can be used to find more accurate latent characteristics of the objects, such as the preferences of users on certain class of products, and to alleviate the data sparsity and data missing problem due to the successful generalization performance. The cluster structure based block term $\mathbf{z}_{i} \mathbf{C} \mathbf{z}_{j}^{T}$ can be capable of learning both dense and sparse clusters at the same time in the relational data, and providing the interpretable latent cluster assignments for each object, which makes the link prediction for some unobserved pairs feasible when these pairs have been clustered in certain block. Hence, the integration of multiple effects makes our proposed model better generalization in prediction and understanding in local structure.

To make our proposed model more accurate, we can impose some prior distributions on the latent factors. For example, the latent cluster indicator vector for each object can be generated based on Multinomial distribution. We can also put normal priors on the latent feature factors and the block matrix as follows:
\begin{equation*}
\begin{aligned}
& p (\mathbf{U}) = \prod_{i}  \mathcal{N}(\mathbf{u}_{i} | 0, \Lambda^{-1}_{U}\mathbf{I}); \quad  p (\mathbf{V}) = \prod_{i}  \mathcal{N}( \mathbf{v}_{i} | 0, \Lambda^{-1}_{V}\mathbf{I}) \\
& p (\mathbf{C}) = \prod_{k,l}  \mathcal{N}( {C}_{kl} | 0, \Lambda^{-1}_{C})
\end{aligned}
\end{equation*}

Based on the above descriptions, we can summarize the joint distribution of our proposed LFBM model as follows:
\begin{equation}
\label{jointob}
\begin{aligned}
p ( \mathbf{S},\mathbf{H},\Theta) & = \prod_{i,j} p( {S}_{ij} \mid H_{ij}) p({H}_{ij} \mid \Theta )\\
& = \prod_{i,j} p( {S}_{ij} \mid H_{ij})  p(U_{i}) p(V_{j}) p ( \beta ) p( \epsilon) \prod_{k,l} p(C_{kl})
\end{aligned}
\end{equation}
where $\Theta$ denotes the normal prior parameters for the latent factors in our model. Then we can build our model on the observed data, and provide the parameter estimation for learning the latent factors in the next section.

\section{Model Inference}
Since the inference task is to estimate the latent factors and parameters in the model, we need to maximize the posterior distribution $p(U,V,C,Z,\beta | S, \Theta)$ given the observed data and the Bernoulli-Logistic style model likelihood. Markov Chain Monte Carlo (MCMC) algorithm has been adopted in such latent variable models to compute the posterior distributions, however, it always costs expensive computation and converges in a slow rate. Therefore, in our model we employ the \emph{Maximum A Posterior} (MAP) strategy to learn the latent factors and model parameters. Then the model inference problem can be formulated as to maximize the log-posterior probability as follows£º
\begin{equation}
\label{logpost}
\begin{aligned}
& L(U,V,C,Z) = \sum\limits_{i,j}  W_{ij} \big ( {S}_{ij} H_{ij} - \log ( 1+ \exp(H_{ij}) ) \big ) \\
& - \frac{\Lambda^{-1}_{U}}{2}tr(UU^{T}) - \frac{\Lambda^{-1}_{V}}{2}tr(VV^{T}) - \frac{\Lambda^{-1}_{C}}{2} tr(CC^{T}) + E \\
\end{aligned}
\end{equation}
where $E$ is a constant independent of the model parameters. Since the optimization objective (\ref{logpost}) is not jointly convex in all the model parameters and latent factors, a globally optimal solution is nontrivial to obtain, we resort to the alternating projection algorithm to learn the latent factors by fixing all but one factor in the objective function and updating the free factor by the gradient based method. For instance, to learn the latent feature factor for one object $\mathbf{u}_{i}$ with all others fixed, we get the gradient equation as follows:
\begin{equation}
\label{du}
\begin{aligned}
\nabla L(\mathbf{U}) = \frac{\partial L(\mathbf{U})}{\partial \mathbf{u}_{i}} =  \big( W_{i\cdot} \odot (S_{i\cdot} - H_{i\cdot}) \big ) V - \Lambda^{-1}_{U}\mathbf{u}_{i}
\end{aligned}
\end{equation}

Then by setting the gradient (\ref{du}) equals to zero we can obtain the update equation for the latent feature factor $\mathbf{u}_{i}$. However, by (\ref{du}) it is intractable to derive the closed-form iterative update rules for these latent factors; even with the Newton based method, due to the complexity of logistic-log-partition (LLP) function $llp(x)= \log (1+ \exp (x)) $ in (\ref{logpost}), the computation complexity for Hessian matrix with respect to latent factor is cubic in the number of the model parameters.

Thus, we construct an optimization transfer algorithm based on the Generalized Expectation-Maximization (EM) method \cite{Lang:mm} to alleviate the model complexity in optimization. During the optimization procedure, we employ the auxiliary function approach commonly used in EM-style algorithms to form the minorizing function as a concave lower bound of the objective function in the E-step, and then maximizing the minorizing function in the M-step alternately, which is so-called Minorization-Maximization (MM) algorithm.

\subsection{Generalized E-Step}
In the Generalized E-Step, for learning the latent factors in the model, we need to derive the minorizing function by aid of the auxiliary function for the objective \cite{lee:nmf}.
\begin{definition}
Given the objective function $L( \Omega )$ in Equation (\ref{logpost}) \footnote{Here $\Omega$ represents a latent factor while the others are fixed.}, $Q(\Omega, \bar{\Omega})$ is an auxiliary function for $L( \Omega )$ if the conditions
\begin{equation*}
\begin{aligned}
(i) \quad L( \Omega ) = Q(\Omega, \Omega); \quad (ii) \quad L( \Omega ) \geq  Q(\Omega,\bar{\Omega}).
\end{aligned}
\end{equation*}
are satisfied.
\end{definition}

\begin{lemma}
\label{l1}
If $Q(\Omega,\bar{\Omega})$ is an auxiliary function for $L( \Omega )$, then $L( \Omega )$ is non-increasing under the following update:
\begin{equation*}
\begin{aligned}
\Omega^{(t+1)} = \arg \max_{\Omega} Q(\Omega,\Omega^{(t)})
\end{aligned}
\end{equation*}
where $\Omega^{(t)}$ denotes the current estimation of the model parameter and $\Omega^{(t+1)}$ is the new estimation to maximize $Q$.
\end{lemma}

\begin{proof}
$L(\Omega^{(t+1)}) \geq Q(\Omega^{(t+1)}, \Omega^{(t)}) \geq Q(\Omega^{(t)}, \Omega^{(t)}) = L(\Omega^{(t)})$
\qed
\end{proof}

Note that the defined auxiliary function $Q(\Omega, \bar{\Omega})$ is a lower bound of $L( \Omega )$, which can be considered as the minorizing function. For example, to learn the latent feature factor $\mathbf{U}$, we consider the objective function $L(\mathbf{U})$ only with respect to $\mathbf{U}$ while fixing the other factors in (\ref{logpost}), then for the auxiliary function $Q(\mathbf{U}, \mathbf{U}^{(t)})$ of a particular form, we have the following theorem:
\begin{theorem}
\label{theU}
If $K(\mathbf{u}_{i}^{(t)})$ has the form:
\begin{equation*}
\begin{aligned}
K(\mathbf{u}_{i}^{(t)}) = - \frac{1}{8} \sum\limits_{i,j} W_{ij}( V_{j\cdot}^{T} V_{j\cdot})- \Lambda^{-1}_{U}\mathbf{I} \end{aligned}
\end{equation*}
then we have the following auxiliary function $Q(\mathbf{u}_{i}, \mathbf{u}_{i}^{(t)})$:
\begin{equation}
\label{auxiliaryU}
\begin{aligned}
Q(\mathbf{u}_{i}, \mathbf{u}_{i}^{(t)}) & = L(\mathbf{u}_{i}^{(t)}) + ( \mathbf{u}_{i} - \mathbf{u}_{i}^{(t)} )^{T} \nabla L(\mathbf{u}_{i}^{(t)}) \\
& + \frac{1}{2} ( \mathbf{u}_{i} - \mathbf{u}_{i}^{(t)} )^{T} K(\mathbf{u}_{i}^{(t)}) ( \mathbf{u}_{i} - \mathbf{u}_{i}^{(t)} )\\
\end{aligned}
\end{equation}
is an auxiliary function for $L(\mathbf{u}_{i})$.
\end{theorem}

\begin{proof}
See Appendix I for a detailed proof. \qed
\end{proof}

Thus, we can derive the minorizing function $Q(\mathbf{u}_{i}, \mathbf{u}_{i}^{(t)})$ for learning latent feature vector $\mathbf{u}_{i}$, which is also the lower bound of $L(\mathbf{u}_{i})$, then we optimize the latent parameters by maximizing $Q(\mathbf{u}_{i}, \mathbf{u}_{i}^{(t)})$ in the next M-step. Note that this optimization transfer algorithm is similar to Newton's method for maximizing $L(\mathbf{u}_{i})$ by replacing the Hessian at each iteration by the derived matrix $K(\mathbf{u}_{i}^{(t)})$, which needs to be inverted only once, rather than at each iteration.

Similar to the optimization for latent factor $\mathbf{U}$, we can derive the minorizing functions $Q(\mathbf{V}, \mathbf{V}^{(t)})$, $Q(\mathbf{C}, \mathbf{C}^{(t)})$ and $Q(\beta, \beta^{(t)})$ with the specific matrices $K(\mathbf{v}_{i}^{(t)})$, $K(\mathbf{C}^{(t)})$ and $K(\beta^{(t)})$ respectively.

To learn the latent cluster assignment $\mathbf{z}_{i}$ for each object, since each object in the relational data is exclusively assigned to a single latent cluster, we can find the optimal value quite efficiently by maximizing the log-posterior probability as follows:
\begin{equation*}
\label{clusterupdate}
\begin{aligned}
\mathbf{z}_{i} = \arg \max_{k}  \big ( \sum\limits_{j} W_{ij} ( {S}_{ij} H_{ij} - \log ( 1+ \exp(H_{ij}) )) \big ) + \hat{E} \\
H_{ij} = \beta^{T} \mathbf{x}_{ij} + \mathbf{u}_{i} \mathbf{v}_{j}^{T} + \sum\limits_{k,l=1}^{K}( {z}_{ik} {C}_{kl} {z}_{jl})
\end{aligned}
\end{equation*}
where $\hat{E}$ is a constant independent of latent cluster assignment $\mathbf{z}_{i}$. Moreover, since the latent cluster assignments $\mathbf{z}_{i}$ and $\mathbf{z}_{j}$ are exclusively binary-valued, by which the probability of a relation between objects $i$ and $j$ could directly be mapped to the corresponding block value ${C}_{kl}$, we can convert the term $\mathbf{z}_{i} \mathbf{C} \mathbf{z}_{j}^{T}$ to the form of $\mathbf{C}^{T} \tilde{\mathbf{Z}}_{ij}$, where $\tilde{\mathbf{Z}}_{ij}$ is the kronecker product of $\mathbf{z}_{i}$ and $\mathbf{z}_{j}$. Then the latent block matrix $\mathbf{C}$ can be learnt more efficiently as in the generalized linear model.

\subsection{Generalized M-Step}
In the M-step, we can optimize the latent factors and model parameters by maximizing the obtained minorization functions, which is the quadratic functions of one latent factor while fixing the others.

First, we derive the update rule for latent feature vector $\mathbf{u}_{i}$ for each object. Based on the Theorem 1, optimizing $\mathbf{u}_{i}$ is equivalent to deriving the Newton step for $\mathbf{u}_{i}$ in $Q(\mathbf{u}_{i}, \mathbf{u}_{i}^{(t)})$ as follows:
\begin{equation}
\label{updateu}
\begin{aligned}
\mathbf{u}_{i}^{(t+1)} & = \mathbf{u}_{i}^{(t)} - \eta \cdot \nabla Q(\mathbf{u}_{i}, \mathbf{u}_{i}^{(t)}) [\nabla^{2} Q(\mathbf{u}_{i}, \mathbf{u}_{i}^{(t)}) ]^{-1} \\
& = \mathbf{u}_{i}^{(t)} - \eta \cdot \nabla L(\mathbf{u}_{i}^{(t)}) [K(\mathbf{u}_{i}^{(t)})]^{-1} \\
\end{aligned}
\end{equation}
where $\eta$ can be set by using the Armijo's rule \cite{Singh:cmf}.

Then, we could also derive the update rules for latent factors $\mathbf{V}$ and $\beta$ similarly by using the corresponding minorizing functions $Q(\mathbf{v}_{i}, \mathbf{v}_{i}^{(t)})$ and $Q(\beta, \beta^{(t)})$ respectively.

For the latent factor $\mathbf{C}$ \footnote{We first convert $\mathbf{C}$ to \emph{vector} $\vec{\mathbf{C}}$}, we could also convert the minorizing function $Q(\vec{\mathbf{C}}, \vec{\mathbf{C}}^{(t)})$ by employing the kronecker product of latent cluster assignments of objects, which makes the optimization of $\mathbf{C}$ easily as follows:
\begin{equation}
\label{updatec}
\begin{aligned}
\vec{\mathbf{C}}^{(t+1)} = \vec{\mathbf{C}}^{(t)} - \eta \cdot \nabla L(\vec{\mathbf{C}}^{(t)}) [K(\vec{\mathbf{C}}^{(t)})]^{-1}
\end{aligned}
\end{equation}

Based on the above updating rules of latent factors, the objective function will monotonically increase; and after combining the Minorization and Maximization step, the learning procedure of our proposed model can converge to a local maximum.

\subsection{Model Complexity}
Since the alternating iterative updating rules are employed to learn the latent factors, in the generalized EM algorithm learning $U$ and $V$ requires a computation time of $O(Nd^{2})$ in each iteration, where $d$ is the dimension of the latent feature factor and $N$ denotes the number of observations in the data. Learning the regression coefficients $\beta$ requires a computation time of $O(Nm^{2})$ in each iteration, where $m$ denotes the dimension of the side-information vector, and learning the latent block factor $C$ requires a computation time of $O(Nk^{2})$ in each iteration. Since each object in the model is assigned to a single latent cluster, the computation complexity of the cluster assignments requires only $O(Nk)$ per iteration. Thus, assuming a few number of iterations, the overall learning algorithm only requires a computation time which is linear to the number of observations, which provides the possibility for handling large scale relational data.

\section{Experiments}
In this section, we demonstrate how our proposed model performs on both synthetic data and real world relational data compared with alternative methods.

\subsection{Experimental Setup}
In the experiments, we compare the following network modeling methods in terms of two tasks.
\begin{itemize}
    \item NMF model \cite{Xu:nmfc} indicates the Nonnegative Matrix Factorization model, which is used for data clustering by learning the latent feature factors in the latent semantic space.
    \item MMSB model \cite{ai:mmsb} indicates Mixed Membership Stochastic Blockmodel, which considers to use only relation matrix to discover the latent cluster assignments of each object.
    \item MLFM model \cite{Hoff:MLFM} indicates the Multiplicative Latent Factor Model, which learns latent feature factors from the relational data by assuming under the Bernoulli distribution.
    \item GLFM model \cite{Li:GLFM} is the Generalized Latent Factor Model, which generalizes the MLFM model to learn more accurate latent feature factors.
    \item LFBM model is our proposed Latent Factor Block Model.
\end{itemize}

To examine how well the compared models perform on the relational data, we evaluate two related tasks: link prediction task and cluster analysis task:
\begin{itemize}
  \item From the \emph{link prediction} task, we can check the generalization and prediction performance of the compared models when the relational data is sparse and noisy. Without loss of generality, we use only the relation matrices with the objects' identifiers without considering the side-information, and we generate random train/test data splits from the relational data, and compute the average AUC (the Area Under the receiver operating characteristic Curve) values against the ground truth test data.
  \item From the \emph{cluster analysis} task, we can check the difference about the ability in discovering the latent cluster assignments by our proposed model and other models. In this task, we consider the dataset with content information as well as the relational structure. We use NMI (Normalized Mutual Information) as the metric to measure the clustering accuracy of the models, which is a standard way to measure the cluster quality. For the compared models, we set the number of latent clusters to the ground-truth number of class labels in the data.
\end{itemize}

\subsection{Experiment on Synthetic Data}
We first use synthetic binary relational data to examine our models. We generate the synthetic data matrix with the number of $200$ objects with noises, representing a network with three three clusters as shown in Figure 3(a). Specifically, in the first two clusters the objects are fully connected (i.e, the corresponding sub-matrices are dense clusters), while within the third cluster the objects are not inter-connected (i.e, the corresponding sub-matrix is sparse) but connected to the objects in the first cluster.

To check how well the proposed models work on fitting to the relational data, we conduct the task of reconstructing the original data by our proposed models. Figure 3 demonstrate the reconstructions by fitting different models to the data. From the results, we can find that MMSB, GLFM and our proposed LFBM model can reveal clearer structures than NMF and GLFM models, which indicates the limitations of factorization-style only based latent factor models. For fair comparison, the hyper-parameters are selected from a wide range and the best selection are reported. We set the number of latent clusters to $k = 3$, and set the dimension of latent feature factors to $d = 2$ for all the models, and the $\eta$ is set to $0.2$, $\Lambda_{U}$, $\Lambda_{V}$ and $\Lambda_{C}$ are set to 1. Note that how to choose the optimal number of the latent clusters is beyond the scope of this paper, and will the future work.

Then we check the performances on link prediction task by using the compared models. We randomly choose 90\% of the relation data as the training data by setting the weight matrix $W$, and leave the other as missing data for test. For the experiments we evaluate the models by repeating the process five times and report the average results. Figure 2 shows the ROC curve and the average AUC performances of different models for the link prediction task. It can be observed that the best performing method among all the models is our proposed LFBM model, which indicates that it can efficiently provide better generalization and predictive performance compared to the MMSB model based only cluster structure, and also beat the other factorization based models (e.g., NMF and MLFM) due of its flexibility to discover special clusters and integrate the benefits from latent block structure.
\begin{figure}[t]
\centering
\includegraphics[width=2.5in]{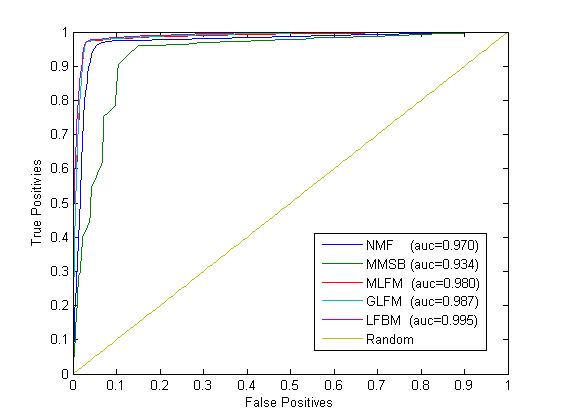}
\caption{ROC curve with different methods for Link Prediction task on synthetic data. The average AUC values are also demonstrated. }
\end{figure}

\begin{figure*}[t]
\centering
\subfigure[]{\includegraphics[width=1.9in]{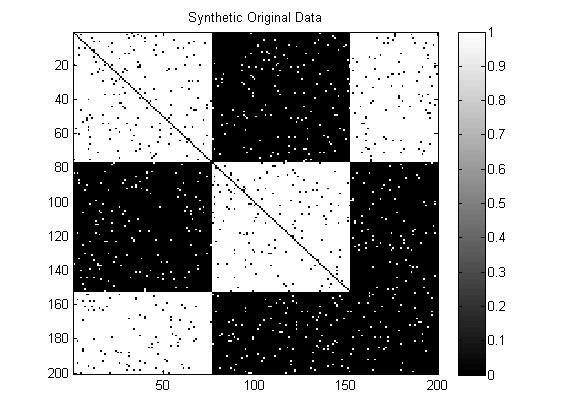}}
\subfigure[]{\includegraphics[width=1.9in]{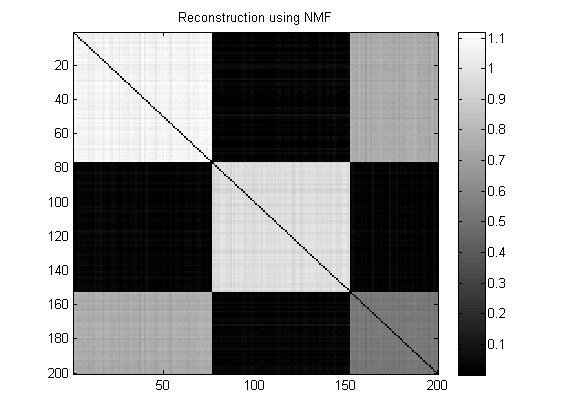}}
\subfigure[]{\includegraphics[width=1.9in]{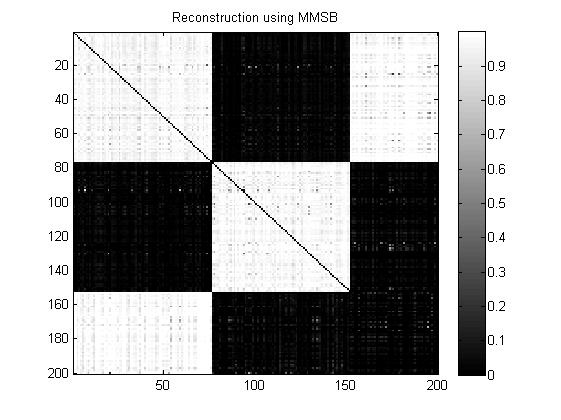}}
\subfigure[]{\includegraphics[width=1.9in]{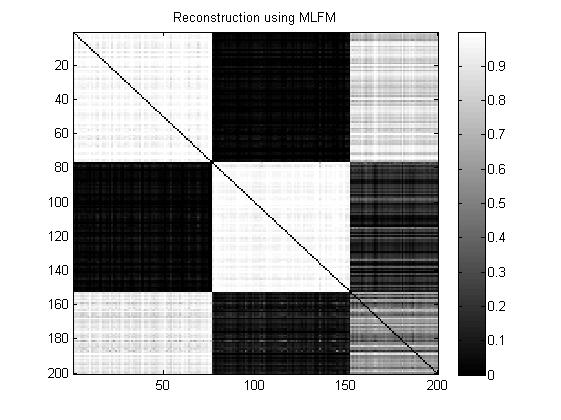}}
\subfigure[]{\includegraphics[width=1.9in]{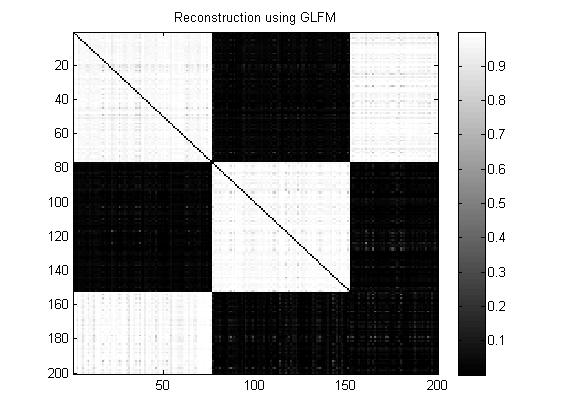}}
\subfigure[]{\includegraphics[width=1.9in]{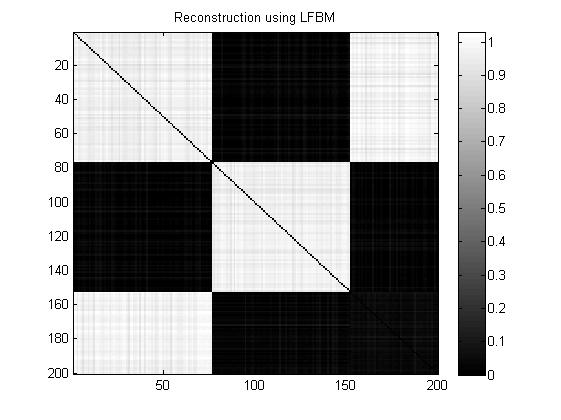}}
\caption{The synthetic data consists of three clusters. (a) Original data with noises. (b) Reconstruction data matrix using NMF. (c)Reconstruction data matrix using MMSB (d)Reconstruction data matrix using MLFM. (e)Reconstruction data matrix using GLFM. (f)Reconstruction data matrix using LFBM. }
\end{figure*}

We also conduct the experiments to check the latent cluster assignments learned by each model. We use the NMI between the resulting cluster labels and the ground truth labels to measure the cluster quality. For our proposed model LFBM and MMSB model, the resulting cluster labels can be obtained directly from the latent cluster assignment factors in the models, while for NMF, MLFM and GLFM models, we use latent feature factors $U$ to determine the cluster label of each data point, and assign object $i$ to latent cluster $k$ if $k = \arg \max_{d} \mathbf{u}_{id}$. Then the average performance scores are reported in Figure 4. We observe that our proposed model has the NMI score comparable with the MMSB model, which consider only the latent structure in the data, rather better than the other models based on latent feature factors, which proves the flexibility and generality of LFBM model in analyzing the relational data.

\begin{figure}[t]
\centering
\includegraphics[width=2.3in]{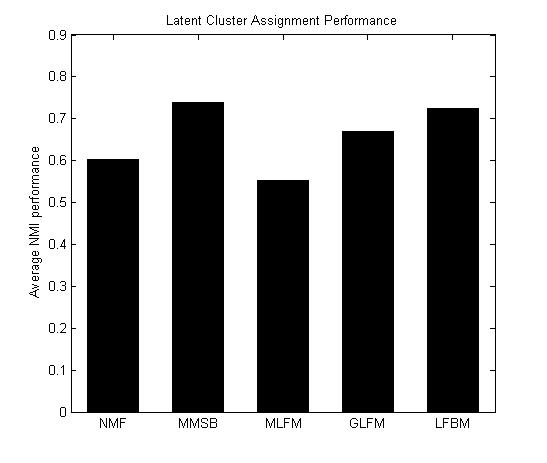}
\caption{The NMI performance for the latent cluster assignments of each object using different models. }
\end{figure}

\begin{table*}[t]
 \caption{Average AUC performances on LiveJournal and Coauthor datasets using different models when varying the dimension of latent feature factors. Best results are in bold.}
\label{data}
 \center \small
  \begin{tabular*}{\linewidth}{@{\extracolsep{\fill}}|l|c|c|c|c||c|c|c|c|}
  \hline \hline
          & \multicolumn{4}{c||}{LiveJournal} & \multicolumn{4}{c|}{Coauthor} \\
  \hline
                & d=10   & d=20   & d=30   & d=40   & d=20   & d=30   & d=40   & d=50 \\  \hline
  \hline
          NMF   & 0.7370 & 0.7468 & 0.7579 & 0.7612 & 0.6801 & 0.6973 & 0.7128 & 0.7212 \\
          MMSB  & 0.6512 & 0.6512 & 0.6512 & 0.6512 & 0.6099 & 0.6099 & 0.6099 & 0.6099 \\
          MLFM  & 0.7804 & 0.8023 & 0.8103 & 0.8115 & 0.7345 & 0.7432 & 0.7486 & 0.7521 \\
          GLFM  & 0.8115 & 0.8319 & 0.8401 & 0.8483 & 0.7676 & 0.7801 & 0.7852 & 0.7967 \\
          LFBM  & \textbf{0.8568} & \textbf{0.8720} & \textbf{0.8793} & \textbf{0.8805} & \textbf{0.8029} &\textbf{ 0.8105} & \textbf{0.8213} & \textbf{0.8232} \\
          \hline \hline
  \end{tabular*}
\end{table*}

\subsection{Experiment on Real World Datasets}
We compare the performance of the proposed models on real world datasets for different tasks. We also report results for a varying number of latent features $d$ and a varying number of latent clusters $k$. In general increasing the number of parameters improves the performance, however, there is a compromise between complexity and performance in the models.

\begin{figure}[t]
\centering
\includegraphics[width=2.3in]{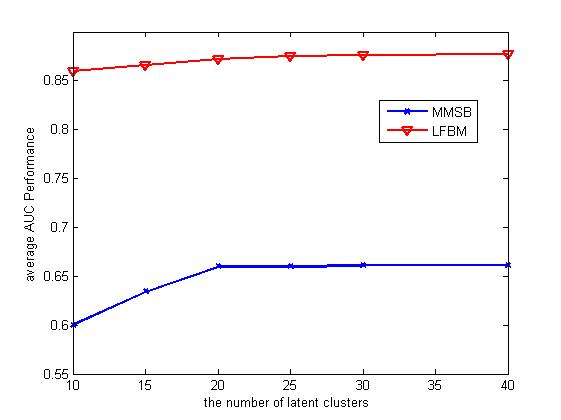}
\caption{The average AUC values for MMSB and LFBM models when varying the number $k$ of latent clusters on LiveJournal dataset with $d = 20$.}
\end{figure}

\subsubsection{Case Study 1: Link Prediction in Social Networks}
We use two social network datasets for the link prediction experiments. The first dataset is a online social network data from the LiveJoural website. The LiveJournal dataset \cite{Freddy:SC} contains binary social friendship between users from the website, consisting of $3,773$ users and $209,832$ social links. The second dataset is a coauthorship dataset from arXiv archive \cite{Gao:grjmf}. The Coauthorship dataset contains binary coauthorship between two authors to indicate whether they have written one paper together, consisting of $2,403$ authors and $21,397$ coauthorship. In both datasets, we do not use any side information, and only learn the latent feature factors and latent block structural factors. For each of these datasets, we randomly choose 80\% of the relation matrix entries for training, then we assume the remaining 20\% of entries as missing and predict them for testing. The experiment is repeated 10 times. We evaluate all the models by average AUC values averaged over five times. We set the number of latent clusters to $k = 20$, and set the dimension of latent feature factors to $ d = 20 $  for the LiveJournal dataset and set $d=40$ for Coauthor dataset in all the models, and the $\eta$ is set to $0.25$, $\Lambda_{U}$, $\Lambda_{V}$ and $\Lambda_{C}$ are set to $0.5$.

Experimental results are shown in Table 1. We find that our proposed LFBM model outperforms all the other models in both datasets, which suggests that integrating the effects of learning latent feature information and modeling latent block structure leads to better performance compared to the models that do not consider both these effects simultaneously. Taking the LiveJournal dataset as the example, comparing to the local structure based MMSB model, LFBM gains much higher improvement, which proves the excellent predictive performance of latent feature factorization based methods on link prediction task. While with respect to the factorization based models (e.g. NMF, MLFM, GLFM) which only learn the latent feature factors for learning, LFBM model improves the performance about $10\% \sim 20\%$, indicating that exploiting the latent structure in link prediction task is quite effective to achieve better performance. Specifically, it is also important to note that our proposed LFBM model and other compared models are under Bernoulli distribution rather than the NMF model under Normal distribution, which should remind us that although Normal distribution is most popular in network modeling, LFBM model under Bernoulli or other possible distributions may be more suitable in link prediction applications due to the statistical properties in relational data.

Moreover, Table 1 also shows the performance evolution when the dimension $d$ of the latent feature factors varies in a wide range. In this range, the higher the dimension of the latent feature representation, the better the predictive performance is. In order to achieve a compromise between model complexity and performance we then fix $d$ to $20$ in LiveJournal dataset and fix $d$ to $40$ in Coauthor dataset in the experiments, since the relation matrix from LiveJournal dataset is a little denser social network than from the Coauthor dataset.

\begin{figure}[t]
\centering
\includegraphics[width=2.3in]{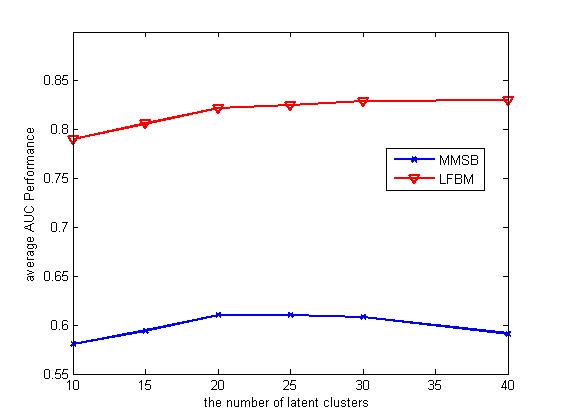}
\caption{The average AUC values for MMSB and LFBM models when varying the number $k$ of latent clusters on Coauthor dataset with $d = 40$. }
\end{figure}

To consider the effect of latent cluster structure in link prediction task, we also vary the number of latent clusters $k$ from 10 to 40 in both datasets and compare with the related models, i.e. MMSB and LFBM. Since for these real world datasets, we do not have the true cluster label for each object in the datasets, we cannot evaluate the estimation accuracy of latent cluster assignments and thus we only look at the average AUC performance on the prediction accuracy for the missing data. From Figure 5 and Figure 6 we observe that LFBM highly outperforms MMSB model with respect to all the varying number of latent clusters in both datasets respectively, which again proves the benefit of simultaneously incorporating both latent feature factorization and latent cluster information for constructing predictive models for relational data.

\subsubsection{Case Study 2: Cluster Analysis in Citation Networks}
In this study, we consider the paper citation dataset for the cluster analysis task. The used experiment data is Cora dataset, which contains $2708$ papers from the 7 subfields (i.e., probabilistic methods, case-based reasoning, genetic algorithms, neural networks, reinforcement learning, rule learning and theory) and $5429$ citations between these  papers. Each paper has side information that consists of a binary word vector indicating the absence/presence of the corresponding word from a dictionary. In the cluster detection task, for MMSB and LFBM models, we set the number of latent clusters to the ground-truth number of class labels in the data, while for other factorization-style models (e.g. NMF, MLFM, GLFM), we choose the dimension of the latent feature factors to be the number of class labels to simulate the latent cluster assignment for fair comparison \footnote{Here we use the different way as in \cite{Li:GLFM} which adopts k-means to perform clustering based on the normalized latent factors.}. Note that the correlations between the dimension of latent feature factor and the number of latent cluster will be discovered in our future research.

In the experiments, we evaluate the latent cluster assignment performance in terms of NMI score for all the models. For that, we set $\eta = 0.2$, and $\Lambda_{U}$, $\Lambda_{V}$ and $\Lambda_{C}$ are set to $0.5$, the dimension of latent feature factors in LFBM model is set to $20$. The results are reported in Table 2. From the results we can find that LFBM achieves the best performance among all the models, which indicates that our model can reveal more clearer cluster structure by simultaneous inclusion of latent feature and latent block factors compared to MMSB model and other latent feature factor based models. Specifically, the latent cluster structures obtained by LFBM model after adjusting for latent feature factors are more informative than by MMSB model. Moreover, considering the paper citation network containing sparse clusters in which papers may be clustered based on their content information even without the citations as well as dense clusters where the citations often exist among papers, GLFM and MMSB model can only find the dense ones, while our proposed LFBM model is flexible to reveal the mixed structures and obtain much better cluster accuracy.

Since LFBM model construct both latent feature factors and latent structure factors to model the relational data, we also want to check the performance evolution when the number of latent feature factor $d$ varies. Figure 7 shows the NMI performance in the experiments, from which we can find that the higher the dimensionality of the latent factor, the better the performance of LFBM model is until some extent. Due to the compromise between complexity and performance in the experiments, we select $d = 20$ in the experiments.

\begin{table}[t]
 \caption{ NMI performances on Cora dataset using different models. Best results are in bold and second best in italic.}
\label{data}
  \center\small
  \begin{tabular*}{\linewidth}{@{\extracolsep{\fill}}|l||c|c|c|c|c|}
  \hline\hline
                  &  NMF   & MMSB     & MLFM     & GLFM    & LFBM \\  \hline
    Cora          & 0.2839   & 0.2005     & 0.3302   & \emph{0.4582}  & \textbf{0.5195} \\
    \hline \hline
  \end{tabular*}
\end{table}

\begin{figure}[t]
\centering
\includegraphics[width=2.4in]{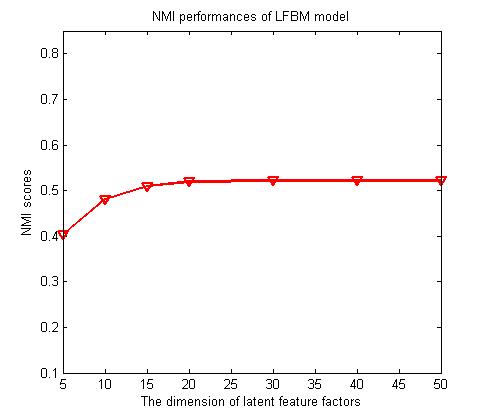}
\caption{NMI performance of LFBM model when varying the dimension of latent feature factor $d$. }
\end{figure}

\section{Related work}
For the problem of modeling relational data, there are various approaches, which can be classified into latent feature based and latent cluster structure based models.

\textbf{Latent Feature Based Models}: Latent Feature Models are based on matrix or tensor factorization, which learn a distributed representation for each object and each relation, and then make predictions by taking appropriate inner products. Their strength lies in the relative ease of their continuous optimization, and in their excellent predictive performance. The representative model is the Multiplicative Latent Factor Model (MLFM) \cite{Hoff:MLFM} and the Generalized Latent Factor Model (GLFM) \cite{Li:GLFM}. For example, MLFM includes both the latent class model and the latent distance model as special cases, can somewhat capture both homophily and stochastic equivalence in networks. However, this kinds of latent factor models are often hard to understand and to analyze the learned latent structure. There is a log-linear model with latent features for dyadic prediction in relational data \cite{Menon:LLM}.

\textbf{Latent Structure Based Models}: The latent structure based models provide building latent blocks for complex networks and allow us to understand and predict unknown interactions between network nodes. For example, stochastic blockmodels \cite{tom:sbm} adopt mixture models for relational data. In this model, each node is sampled from a cluster based on a multinomial distribution. To allow a node belonging to multiple groups, Airoldi et al. \cite{ai:mmsb} developed mixed membership stochastic blockmodels, which use a latent Dirichlet allocation prior to model latent membership variables. Agarwal et al. in \cite{Agarwal:PDLFM} proposed the Predictive Discrete Latent Factor (PDLF) model to predict large scale dyadic response variables. The model simultaneously incorporates the effect of covariates and estimates local structure that is induced by interactions among the dyads through a discrete latent factor model. There is also similar work in \cite{Sutskever:BCTF}. Another related research is the relational clustering. Long et al. \cite{Long:SCC} and \cite{Long:PF} has proposed a general model for relational clustering based on symmetric convex coding.

An alternative approach to modeling dependencies among relational data is the use of relational Gaussian process model \cite{YU:GPLP}, which has been successfully applied to a variety of relational learning problems. Essentially Yu et al. in \cite{YU:GPLP} use a linear covariance function and Gaussian likelihood in their relational GP model. Yan et al. \cite{Yan:GPBM} proposed the sparse matrix-variate Gaussian process blockmodel to generalize the bilinear generative models to handle nonlinear network interactions.

\section{Conclusion and Future work}
In this paper, we have addressed the problem of modeling relational data. For that we proposed a novel model that simultaneously incorporates the effects of latent feature factors and the impacts from the latent block structure in the network. The model can collectively capture globally predictive intrinsic properties of objects and discover the latent block structure, which shows the success of the coupled benefits of latent feature factorization based approaches and latent class based approaches in providing better predictive performance and much clearer latent block structure. We also employed an efficient optimization transfer algorithm to alleviate the model complexity. Extensive experiments on the synthetic data and several real world datasets suggest that our proposed LFBM model outperforms the other state of the art approaches in the evaluated tasks for modeling the relational data.

There are still directions remaining to be explored. First it would be interesting to investigate how to automatically choose the number of latent feature factors and the number of latent clusters from the data, and reveal the correlations between them. Second, the model learning in this paper employs the MAP strategy, however it would be promising to use a Bayesian approach to involve the marginalization of latent factors and model parameters.

\bibliographystyle{abbrv}
\bibliography{llncs}  

\appendix
\section{Appendix I: Proof of Theorem 1}
To prove Theorem 1, we make use of the auxiliary function $Q(\mathbf{u}_{i}, \mathbf{u}_{i}^{(t)})$, and derive the proof as follows:
\begin{proof}
Given the form of $Q(\mathbf{u}_{i}, \mathbf{u}_{i}^{(t)})$, the auxiliary function should satisfy the required conditions in Definition 1. For the first condition, it is easy to observe that $L(\mathbf{u}_{i}) = Q(\mathbf{u}_{i}, \mathbf{u}_{i})$. For the second condition $L(\mathbf{u}_{i}) \geq  Q(\mathbf{u}_{i}, \mathbf{u}_{i}^{(t)})$, since the objective function $L(\mathbf{u}_{i})$ is convex with bounded curvature, we first derive the Taylor series of  $L(\mathbf{u}_{i})$ with respect only to $U$ as follows:
\begin{equation*}
\label{tayloru}
\begin{aligned}
L(\mathbf{u}_{i}) & = L(\mathbf{u}_{i}^{(t)}) + ( \mathbf{u}_{i} - \mathbf{u}_{i}^{(t)} )^{T} \nabla L(\mathbf{u}_{i}^{(t)}) \\
& + \frac{1}{2} ( \mathbf{u}_{i} - \mathbf{u}_{i}^{(t)} )^{T} \nabla^{2} L(\mathbf{u}_{i}^{(t)}) ( \mathbf{u}_{i} - \mathbf{u}_{i}^{(t)} )\\
\end{aligned}
\end{equation*}
where $\nabla^{2} L(\mathbf{u}_{i}^{(t)})$ can be computed as follows:
\begin{equation*}
\label{hessianu}
\begin{aligned}
\nabla^{2} L(\mathbf{u}_{i}^{(t)}) =  - \frac{1}{2} \sum\limits_{i,j} \{W_{ij} \sigma(H_{ij})(1- \sigma(H_{ij})) V_{j\cdot}^{T}V_{j\cdot} \} - \Lambda^{-1}_{U}\mathbf{I}
\end{aligned}
\end{equation*}

Based on the inequality $\{ \sigma(x)(1- \sigma(x) \leq \frac{1}{4} \}$, it is trivial to observe that $K(\mathbf{u}_{i}^{(t)})$ and $\nabla^{2} L(\mathbf{u}_{i}^{(t)}) - K(\mathbf{u}_{i}^{(t)})$ are positive semi-definite matrices.

Comparing the forms of $L(\mathbf{u}_{i})$ and $Q(\mathbf{u}_{i}, \mathbf{u}_{i}^{(t)})$:
\begin{equation*}
\label{auxiliaryU}
\begin{aligned}
Q(\mathbf{u}_{i}, \mathbf{u}_{i}^{(t)}) & = L(\mathbf{u}_{i}^{(t)}) + ( \mathbf{u}_{i} - \mathbf{u}_{i}^{(t)} )^{T} \nabla L(\mathbf{u}_{i}^{(t)}) \\
& + \frac{1}{2} ( \mathbf{u}_{i} - \mathbf{u}_{i}^{(t)} )^{T} K(\mathbf{u}_{i}^{(t)}) ( \mathbf{u}_{i} - \mathbf{u}_{i}^{(t)} )\\
\end{aligned}
\end{equation*}
we can observe that $L(\mathbf{u}_{i}) \geq  Q(\mathbf{u}_{i}, \mathbf{u}_{i}^{(t)})$.

Thus $Q(\mathbf{u}_{i}, \mathbf{u}_{i}^{(t)})$ is the auxiliary function for $L(\mathbf{u}_{i})$, which is also the lower bound of $L(\mathbf{u}_{i})$. Theorem 1 is proved. \qed
\end{proof}

\end{document}